\documentclass{article}
\usepackage{amsmath,amsthm}
  \usepackage{paralist}
  \usepackage{graphics} 
  \usepackage{epsfig} 
\usepackage{graphicx}  \usepackage{epstopdf}
 \usepackage[colorlinks=true]{hyperref}
\usepackage{authblk}

\usepackage{amssymb}
\usepackage{graphicx, color}
\usepackage{algorithm,algorithmic}

\hypersetup{urlcolor=blue, citecolor=red}

\newtheorem{prop}{Proposition}
\newtheorem{lem}[prop]{Lemma}
\newtheorem{thm}[prop]{Theorem}
\newtheorem{cor}[prop]{Corollary}

\newtheorem{defi}[prop]{Definition}

\newtheorem{ex}[prop]{Example}
\newtheorem{rem}[prop]{Remark}

\newcommand{\F}{\mathbb{F}}

\newcommand{\N}{\mathbb{N}}

\newcommand{\rk}{\mathrm{rank}}

\newcommand{\rs}{\mathrm{rs}}

\newcommand{\lm}{\mathrm{lm}}
\newcommand{\lpos}{\mathrm{lpos}}

\newcommand{\qdeg}{\mathrm{qdeg}}

\newcommand{\Lp}{\mathcal{L}_q(x,q^m)}

\title{Gabidulin Decoding via Minimal Bases of Linearized Polynomial Modules}

\date{}

\author[1]{Anna-Lena Horlemann-Trautmann\thanks{ALHT was partially supported by Swiss National Science Foundation Fellowship no.\ 147304.}}
\author[2]{Margreta Kuijper}

\affil[1]{Laboratory of Algorithmic Mathematics, EPF Lausanne, Switzerland}
\affil[2]{Department of Electrical and Electronic Engineering, University of Melbourne, Australia}

\allowdisplaybreaks

\begin{document}
\maketitle

\begin{abstract}
We show how Gabidulin codes can be decoded via parametrization by using interpolation modules over the ring of linearized polynomials with composition. Our decoding algorithm computes a list of message words that correspond to all closest codewords to a given received word. This involves the computation of a minimal basis for the interpolation module that corresponds to the received word, followed by a search through the parametrization for valid message words.
Our module-theoretic approach strengthens the link between Gabidulin decoding and Reed-Solomon decoding.   
Two subalgorithms are presented to compute the minimal basis, one iterative, the other an extended Euclidean algorithm. Both of these subalgorithms have polynomial time complexity. The complexity order of the overall algorithm, using the parametrization, is then compared to straightforward exhaustive search as well as to chase list decoding. 
\end{abstract}

\section{Introduction}

Over the last decade there has been increased interest in Gabidulin codes, mainly because of their relevance to network coding~\cite{ko08,si08j}. Gabidulin codes are optimal rank-metric codes over a field $\F_{q^m}$ (where $q$ is a prime power). They were first derived by Gabidulin in \cite{ga85a} and independently by Delsarte in \cite{de78}. 
These codes can be seen as the $q$-analog of Reed-Solomon codes, using $q$-linearized polynomials instead of arbitrary polynomials. They are optimal in the sense that they are not only MDS codes with respect to the Hamming metric, but also achieve the Singleton bound with respect to the rank metric and are thus MRD codes. They are not only of interest in network coding but also in space-time coding \cite{lu03}, crisscross error correction \cite{ro91} and distributed storage \cite{si12}. 

The decoding of Gabidulin codes has obtained a fair amount of attention in the literature, starting with work on decoding within the unique decoding radius in~\cite{ga85a,ga92} and more recently~\cite{lo06,ri04p,si10,si11,si09p,wa13d}. If $n$ is the length of the Gabidulin code and $k$ denotes the dimension of the code as a linear space over the field $\F_{q^m}$, the unique decoding radius is given by $(n-k)/2$. 
 Decoding beyond the unique decoding radius was addressed in e.g.\ \cite{
 lo06p,ma12,si10p,wa13a,wa13phd,wa13p}. In this case, one speaks of \emph{list-decoding}, i.e.\ finding all  codewords within a given radius to the received word. A main open question is whether Gabidulin codes can be list decoded efficiently. This paper seeks to contribute to current research efforts on this open question.
In \cite{wa13a} it was shown that, beyond the Johnson radius $n-\sqrt{kn}$, list decoding with a polynomial size list of codewords is not possible. 
This raises the question up to which radius list decoding with a polynomial list size \emph{is} possible.
Recent results \cite{gu13,gu12a} show an explicit construction of rank-metric codes, constructed as subcodes of Gabidulin codes, that can be list-decoded in polynomial time up to a certain radius beyond the unique decoding radius. This motivates further research of what happens in the original Gabidulin setting between the unique and the Johnson radius.

A closely related family of codes is the one of lifted Gabidulin codes \cite{si08j}. These codes are sets of vector spaces
 and can be used for non-coherent (also called random) network coding \cite{ko08}. Unique decoding of lifted Gabidulin codes was investigated in e.g.\ \cite{ko08,si08j}, whereas list-decoding of these codes was studied in \cite{tr13p,xi11}.

Using the close resemblance between Reed-Solomon codes and Gabidulin codes, the paper~\cite{lo06} translates Gabidulin decoding into a set of polynomial interpolation conditions. Essentially, this setup is also used in the papers~\cite{ko08,xi11} that present iterative algorithms that perform Gabidulin list decoding with a list size of 1. In this paper we present an iterative algorithm that bears similarity to the ones in~\cite{ko08,lo06,xi11} but yields {\em all} closest codewords rather than just one. 
The latter is due to our parametrization approach. This approach enables us to strengthen the link between Gabidulin decoding and Reed-Solomon decoding. For the latter a parametrization approach was developed in~\cite{kuijppol_it}. 

The paper is structured as follows. In the next section we present several preliminaries on $q$-linearized polynomials, Gabidulin codes, the rank metric and we recall the polynomial interpolation conditions from~\cite{lo06}. We also detail an iterative construction of the $q$-annihilator polynomial and the $q$-Lagrange polynomial. Section \ref{sec:modules} deals with modules over the ring of linearized polynomials and gives the Predictable Leading Monomial property for minimal bases of these modules. In Section \ref{sec:decoding} we reformulate the Gabidulin list decoding requirements in terms of a module represented by four $q$-linearized polynomials and present the decoding algorithm, which is based on a parametrization using the Predictable Leading Monomial property. For this we present two subalgorithms for computing a minimal basis of the interpolation module. Furthermore, we analyze the complexity of our algorithms. We conclude this paper in Section \ref{sec:conclusion}.

Preliminary short versions of this paper are conference papers~\cite{ku14p} and~\cite{ku14}.


\section{Preliminaries}\label{sec:prelim}

\subsection{$q$-linearized polynomials}

Let $q$ be a prime power and let $\F_q$ denote the finite field with $q$ elements. It is well-known that there exists a primitive element $\alpha$ of the extension field $\F_{q^m}$, such that $\F_{q^m}\cong \F_q[\alpha] $. Moreover, $\F_{q^m}$  is isomorphic (as a vector space) to the vector space $\F_q^m$. 
One then easily gets the isomorphic description of matrices over the base field $\F_q$ as vectors over the extension field, i.e.\ $\F_q^{m\times n}\cong \F_{q^m}^n$. Since we will work with matrices over different underlying fields we denote the rank of a matrix $X$ over $\F_q$ by $\rk_q(X)$.

For some vector $(v_1,\dots, v_n) \in \F_{q^m}^n$ we denote the $k \times n$ \emph{Moore matrix} by
\begin{equation}
M_k(v_1,\dots, v_n) := \left( \begin{array}{cccc}  v_1 & v_2 &\dots &v_n \\ v_1^{[1]} & v_2^{[1]} &\dots &v_n^{[1]} \\ &&\vdots \\  v_1^{[k-1]} & v_2^{[k-1]} &\dots &v_n^{[k-1]} \end{array}\right)   ,\label{eq_moore}
\end{equation}
where $[i]:= q^i$. A \emph{$q$-linearized polynomial} over $\F_{q^m}$ is defined to be of the form
\[f(x) = \sum_{i=0}^{n} a_i x^{[i]}   \quad, \quad a_i \in\F_{q^m} , \]
where $n$ is called the \emph{$q$-degree} of $f(x)$, assuming that $a_n\neq 0$, denoted by $\qdeg (f)$. This class of polynomials was first studied by Ore in \cite{or33}. 
One can easily check that $f(x_1 + x_2)= f(x_1)+f(x_2)$ and $f(\lambda x_1) = \lambda f(x_1)$ for any $x_1,x_2 \in \F_{q^m}$ and $\lambda \in \F_q$, hence the name \emph{linearized}. The set of all $q$-linearized polynomials over $\F_{q^m}$ is denoted by $\Lp$. This set forms a non-commutative ring with the normal addition $+$ and composition $\circ$ of polynomials. 
Because of the non-commutativity, products and quotients of elements of $\Lp$ have to be specified as being ``left" or ``right" products or quotients. To not be mistaken with the standard division, we call the inverse of the composition \emph{symbolic division}. I.e.\ $f(x)$ is symbolically divisible by $g(x)$ with right quotient $m(x)$ if $$ g(x) \circ m(x) = g(m(x)) = f(x).$$
Efficient algorithms for all these operations (left and right symbolic multiplication and division) exist and can be found e.g.\ in \cite{ko08}.

\begin{lem}[cf.\ \cite{li97b} Thm. 3.50]\label{lem:rootspace}
Let $f(x) \in \Lp$ and $\F_{q^s}$ be the smallest extension field of $\F_{q^m}$ that contains all roots of $f(x)$. Then the set of all roots of $f(x)$ forms a $\F_q$-linear vector space in $\F_{q^s}$.
\end{lem}

\begin{defi}
Let $U$ be a $\F_q$-linear subspace of $\F_{q^m}$. We call $\Pi_U(x):=\prod_{g \in U} (x-g)$ the \emph{$q$-annihilator polynomial of $U$}.
\end{defi}

\begin{lem}[\cite{li97b} Thm. 3.52]\label{lem:nullpoly}
Let $U$ be a $\F_q$-linear subspace of $\F_{q^m}$. Then $\Pi_U(x)$ is an element of $\Lp$.
\end{lem}

Note that, if $g_1,\dots,g_n$ is a basis of $U$, one can rewrite 
$$ \Pi_U(x) = \lambda \det(M_{n+1}(g_1,\dots,g_n,x))$$ 
for some constant $\lambda\in\F_{q^m}$. 
Clearly the $q$-degree of this polynomial equals $n$.
%
%
%
%
We also have a notion of $q$-Lagrange polynomial:
\begin{defi}\label{defi:Lagrange}
Let $\mathbf{g} =(g_1,\dots,g_n) \in \F_{q^m}^n$, where $g_1 , g_2 , \ldots , g_n$ are $\F_q$-linearly independent. Let $\mathbf r=(r_1,\dots,r_n) \in \F_{q^m}^n$. Define the matrix $\mathfrak{D}_i(\mathbf g, x)$ as $  M_{n}(g_1,\dots,g_n,x)$ without the $i$-th column.
 We define the \emph{$q$-Lagrange polynomial} corresponding to $\mathbf{g} $ and $\mathbf{r} $ as  
\[\Lambda_{\mathbf g, \mathbf r}(x) := \sum_{i=1}^n (-1)^{n-i}  r_i \frac{\det(\mathfrak{D}_i(\mathbf g, x))}{\det (M_n(\mathbf g))} \quad \in \F_{q^m}[x] .\]
\end{defi}
It can be easily verified that the above polynomial is $q$-linearized and that $\Lambda_{\mathbf g, \mathbf{r}}(g_i) = r_i $ for $i=1,\dots,n$.

Note that, although not under the same name, the previous two polynomials were also defined in e.g.\ \cite{wa13phd}.


In the following we will use matrix composition, which is defined analogously to matrix multiplication:
$$\left[\begin{array}{cc}  a(x) & b(x) \\ c(x) & d(x) \end{array}\right] \circ\left[\begin{array}{cc}  e(x) & f(x) \\ g(x) & h(x) \end{array}\right]:= $$ $$\left[\begin{array}{cc}  a(e(x)) + b(g(x)) & a(f(x))+ b(h(x)) \\ c(e(x))+ d(g(x)) & c(f(x))+ d(h(x)) \end{array}\right] .$$ 

We can recursively construct the $q$-annihilator and the $q$-Lagrange polynomial as follows. Below we use the standard notation $\langle g_1,\dots,g_n\rangle$ for the $\F_q$-linear span of $g_1, g_2, \ldots g_n$.

\begin{prop}\label{prop:Lagrec}
Let $g_1,\dots,g_n\in\F_{q^m}$ be linearly independent over $\F_q$ and let $r_1,\dots,r_n\in\F_{q^m}$. Define
$$\Pi_1(x):=  x^q-g_1^{q-1} x \quad , \quad\Lambda_{ 1}(x):=\frac{r_1}{g_1} x,$$
and for $i=1, \ldots , n-1$
$$ \left[ \begin{array}{cc} \Pi_{i+1}(x) \\ \Lambda_{i+1}(x)  \end{array}\right] :=  \left[ \begin{array}{cc} x^q- \Pi_i(g_{i+1})^{q-1}x & 0 \\ -  \frac{ \Lambda_{ {i}}(g_{i+1}) - r_{i+1}}{\Pi_{i} (g_{i+1})} x  & x  \end{array}\right] \circ  \left[ \begin{array}{cc} \Pi_{i}(x)  \\ \Lambda_{i}(x)   \end{array}\right]  .$$
Then we have $\Pi_i(x) = \Pi_{\langle g_1, g_2, \ldots , g_i \rangle }(x)$ and $\Lambda_i(x) = \Lambda_{(g_1, g_2, \ldots , g_i),(r_1,r_2,\dots, r_{i})}(x)$ for $i=1, \ldots , n$.
\end{prop}
\begin{proof}
We prove this by induction on $i$. The theorem clearly holds for $i=1$. Suppose that the theorem holds for a value of $i$ with $1\leq i < n$. By definition $\Pi_{i+1}(x)= \Pi_i (x)^q - \Pi_i (g_{i+1})^{q-1} \Pi_i (x)$, so that (using the induction hypothesis)  $\Pi_{i+1}(x)$ is a monic $q$-linearized polynomial of $q$-degree $i+1$ such that for $1\leq j \leq i+1$ we have $\Pi_{i+1}(g_j)=0$. It follows that then $\Pi_{i+1}(x)$ must coincide with $\Pi_{\langle g_1, g_2, \ldots , g_{i+1}\rangle }(x)$. 

We next show that the formula for $\Lambda_{i+1}(x)$ yields the $q$-Lagrange polynomial at level $i+1$. 
%
Assume that $\Lambda_i(x)$ is the $q$-Lagrange polynomial at level $i$ and look at $\Lambda_{ {i+1}}(x)$, which is $q$-linearized since $\Lambda_i(x)$ and $ \Pi_{i}(x)$ are $q$-linearized. As $\qdeg( \Pi_{i}(x))= i > \qdeg(\Lambda_i(x))$ it holds that $\qdeg(\Lambda_{i+1}(x)) = i$. Furthermore, because $\Pi_{i}(g_j)=0$ for $j=1,\dots,i$ and $\Lambda_{ {i}}(g_j)=r_j$ for $j=1,\dots,i$ , 
$$\Lambda_{ {i+1}}(g_j) = \Lambda_{ {i}}(g_j) = r_j , \quad \textnormal{ and }$$
$$\Lambda_{ {i+1}}(g_{i+1}) = \Lambda_{ {i}}(g_{i+1}) -  \frac{ \Lambda_{ {i}}(g_{i+1}) - r_{i+1}}{\Pi_{i} (g_{i+1})}   \Pi_{i}(g_{i+1}) = r_{i+1} .$$
Therefore, $\Lambda_{ {i+1}}(x)$ evaluates to the same values as $\Lambda_{(g_1,\dots,g_{i+1}), (r_1,\dots, r_{i+1})}(x)$ 
for $g_1,\dots,g_{i+1}$. Because of the linearity of both these polynomials, they evaluate to the same values for all elements of $\langle g_1,\dots,g_{i+1}\rangle$. Because of the $\F_q$-linear independence of $g_1, g_2, \ldots g_{i+1}$, there are $q^{i+1}$ many values. Since the degree of both polynomials is $q^i < q^{i+1}$, it follows that they must be the same polynomial.
\end{proof}

Let $g_1,\dots,g_n\in\F_{q^m}$ be linearly independent over $\F_q$; as before denote $\mathbf g:=(g_1,\dots,g_n)$. Throughout the remainder of the paper we abbreviate the notation $\Pi_{\langle g_1, g_2, \ldots , g_n\rangle }(x)$ by $ \Pi_{\mathbf g}(x)$. We need the following fact for our investigations in Section \ref{sec:decoding}.

\begin{lem}\label{lem3}
 Let  $L(x) \in \Lp$ be such that $L(g_i)=0$ for all $i$. Then
\[\exists H(x)\in \Lp : L(x) = H(x)\circ \Pi_{\mathbf g}(x)  . \]
\end{lem}
\begin{proof}
We know from Lemma \ref{lem:nullpoly} that $\Pi_{\mathbf g}(x) \in \Lp$. Moreover unique left and right division in $\Lp$ holds, i.e.\ in this case there exist unique polynomials $H(x),R(x)\in \Lp$ such that $L(x) = H(x)\circ\Pi_{\mathbf g}(x) + R(x)$ and $\qdeg(R(x))< \qdeg (\Pi_{\mathbf g}(x)) =n$. Since any $\alpha \in \langle g_1,\dots,g_n\rangle$ is a root of $L(x)$ as well as $\Pi_{\mathbf g}(x)$, they must also be a root of $R(x)$. Hence we have $q^n$ distinct roots for $R(x)$ and $\deg(R)<q^n$, thus $R(x) \equiv 0$ and the statement follows.
\end{proof}


\subsection{Gabidulin codes}

Let $g_1,\dots, g_n \in \F_{q^m}$ be linearly independent over $\F_q$. We define a \emph{Gabidulin code} $C\subseteq \F_{q^m}^{n}$ as the linear block code with generator matrix $M_k(g_1,\dots, g_n)$, as defined in ~(\ref{eq_moore}).        
Using the isomorphic matrix representation, we can interpret $C$ as a matrix code in $\F_q^{m\times n}$.The \emph{rank distance} $d_R$ on  $\F_q^{m\times n}$ is defined by
\[d_R(X,Y):= \rk_q(X-Y) \quad, \quad X,Y \in \F_q^{m\times n} \]
and analogously for the isomorphic extension field representation. 
Then it is clear that the code $C$ has dimension $k$ over $\F_{q^m}$ and minimum rank distance (over $\F_q$) $n-k+1$. One can easily see by the shape of the parity check and the generator matrices that an equivalent definition of the code is
\[C =  \{(m(g_1),\dots,m(g_n))\in \F_{q^m}^n \mid m(x) \in \Lp_{<k}  \} ,\]
where $\Lp_{<k} := \{m(x) \in \Lp \mid \qdeg(m(x)) < k\}$. 
For more information on bounds and constructions of rank-metric codes the interested reader is referred to \cite{ga85a}.

Consider a received word $\mathbf r = (r_1,\dots,r_n) \in \F_{q^m}^n$ as the sum $\mathbf r = \mathbf c + \mathbf e$, where $\mathbf c = (c_1,\dots,c_n)\in C$ is a codeword and $\mathbf e = (e_1,\dots,e_n)\in \F_{q^m}^n$ is the error vector. 
We now recall the polynomial interpolation setup from~\cite{lo06} via a more general formulation in the next theorem.

\begin{thm}\label{thm2}
Let $m(x)\in \Lp, \qdeg(f(x))< k$ and $c_i=m(g_i)$ for $i=1,\dots,n$.
Then $d_R(\mathbf c, \mathbf r) = t$ if and only if there exists a $D(x) \in \Lp$, such that $ \qdeg(D(x))= t$ and
\[D(r_i) = D(m(g_i)) \quad \forall i\in\{1,\dots,n\}.\]
\end{thm}
\begin{proof}
Let $D(x) \in \Lp$ be such that $D(r_i) = D(f(g_i))$ and $\qdeg(D(x)) =t$. This implies that $D(r_i - f(g_i)) = 0$ for all $i$. Define $e_i := r_i - f(g_i)$, then $e_i\in\F_{q^m}$ and every element of $\langle e_1,\dots, e_n\rangle$ is a root of $D(x)$ (see Lemma \ref{lem:rootspace}). Since $D(x)$ is non-zero and has degree $q^t$, it follows that the linear space of roots has $q$-dimension $t$, which implies that $(e_1,\dots,e_n)$ has rank  $t$. This means that the rank distance between $(c_1,\dots,c_n)$ and $(r_1,\dots,r_n)$ is equal to $t$. Thus, one direction is proven.

For the other direction let $(c_1,\dots,c_n), (r_1,\dots,r_n)$ have rank distance $t$, i.e. $(e_1,\dots,e_n) := (c_1-r_1,\dots, c_n-r_n)$ has rank $t$. Then by Lemma \ref{lem:nullpoly} there exists a non-zero $D(x)\in \Lp$ of degree $q^t$ such that $D(e_i)=0$ for all $i$. 
By linearity we get that $D(c_i)=D(r_i)$ for $i=1,\dots,n$. Since we know that $c_i=f(g_i)$, the statement follows.
\end{proof}

\begin{rem}
Theorem~\ref{thm2} states that the roots of $D(x)$ form a vector space of degree $t$ which is equal to the span of $e_1,\dots,e_n$ (for this note that $e_i=m(g_i)-r_i$). This is why $D(x)$ is unique (for given codeword and received word) and is also called the \emph{error span polynomial} (cf.\ e.g.\ \cite{si09}). The analogy in the classical Hamming metric set-up is the \emph{error locator polynomial}, whose roots indicate the locations of the errors, and whose degree equals the number of errors.
\end{rem}

%


\section{Modules over $\Lp$}\label{sec:modules}

As mentioned before, $\Lp$ forms a ring with addition and composition. Hence $\Lp^\ell$ forms a (right or left) module. In this work we will consider $\Lp^{\ell}$ as a left module and investigate its (left) submodules.

In this section, we give some general definitions and results on $\Lp^\ell$ and present the terminology of the Predictable Leading Monomial property. All of these are analogous to the definitions and results for modules over $\F_{q^m}[x]$ (equipped with normal polynomial multiplication) from \cite{al11}, see also the early work by Fitzpatrick~\cite{fi95} and the textbooks \cite{ad94b,co05b}. 
Linearized polynomials belong to the class of skew polynomials, for which the general theory of linear algebra and Gr\"obner bases is well established, see e.g.\ \cite{ab02,be10,ka90}. For reasons of clear exposition, we formulate the results that we need explicitly in terms of rings with composition, more specifically in the language of linearized polynomials. Thus, compared to the $\F_q[x]$-case, multiplication is replaced by composition. 

To avoid confusion, we denote polynomials by $f(x)$, while vectors of polynomials are denoted by $f$. If we need to index polynomials, we use the notation $f_1(x),\dots,f_s(x)$, while for vectors of polynomials we use the notation $f^{(1)}, \dots, f^{(s)}$. 

Elements of $\Lp^\ell$ are of the form 
$$f:=[f_1(x) \; \dots \; f_\ell (x)] = \sum_{i=1}^\ell f_i(x) e_i $$ 
where $f_i(x)=\sum_{j} f_{ij} x^{[j]} \in \Lp$ and $e_1,\dots,e_\ell$ are the unit vectors of length $\ell$. 
Analogous to polynomial multiplication on  $\F_{q^m}[x]^\ell$ we define for $h(x)\in \Lp$ the left operation
\[h(x)\circ f  :=  [h(f_1(x)) \; \dots \; h(f_\ell (x))]= \sum_{i=1}^\ell h(f_i(x)) e_i  .\]
The monomials of $f$ are of the form $x^{[k]} e_i$ for all $k$ such that $f_{ik}\neq 0$.

\begin{defi}
A subset $M\subseteq \Lp^\ell$ is a \emph{(left) submodule} of $\Lp^\ell$ if it is closed under addition and composition with $\Lp$ on the left. 
\end{defi}

\begin{defi}
Consider the non-zero elements $f^{(1)}, \dots, f^{(s)} \in \Lp^\ell$. 
We say that $f^{(1)}, \dots, f^{(s)}$ are \emph{linearly independent} if for any $a_1(x),\dots, a_s(x) \in \Lp$ 
\[\sum_{i=1}^s  a_i(x) \circ f^{(i)} = [\;0 \; \dots \; 0\; ] \quad \implies \quad a_1(x)=\dots =a_s(x) = 0. \]
A generating set of a submodule $M\subseteq \Lp^\ell$ is called a \emph{basis} of $M$ if all its elements are linearly independent.
\end{defi}

One can easily see that 
\[ B= \{x e_1, xe_2 \dots, x e_\ell \}\]
is a basis of  $\Lp^\ell$, thus $\Lp^\ell$ is a \emph{free} and \emph{finitely generated} module.

We need the notion of monomial order for the subsequent results, which we will define in analogy to \cite[Definition 3.5.1]{ad94b}.
\begin{defi}
A \emph{monomial order} $<$ on $\Lp^\ell$ is a total order on $\Lp^\ell$ that fulfills the following two conditions:
\begin{itemize}
\item    $ x^{[k]} e_i  < x^{[j]}\circ (x^{[k]} e_i) $ for any monomial $x^{[k]} e_i \in \Lp^\ell$ and $j\in\mathbb{N}_{>0}$. 
\item    If $x^{[k]} e_i < x^{[k']} e_{i'}$, then $x^{[j]}\circ (x^{[k]} e_i) < x^{[j]}\circ (x^{[k']} e_{i'} )$ for any monomials $x^{[k]} e_{i}, x^{[k']} e_{i'} \in \Lp^\ell$ and $j\in\mathbb{N}_0$. 
\end{itemize}
\end{defi}

We have different choices for monomial orders, of which the following is of interest for our investigations.


\begin{defi}
The \emph{$(k_1,\dots,k_\ell)$-weighted term-over-position monomial order} is defined as 
$$x^{[i_1]} e_{j_1}<_{(k_1,\dots,k_\ell)} x^{[i_2]} e_{j_2} :\iff $$ $$i_1+k_{j_1}<i_2 + k_{j_2} \textnormal{ or }  [i_1+k_{j_1} = i_2+k_{j_2} \textnormal{ and } j_1<j_2 ]  .$$ 
\end{defi}

Note that this monomial order for $\Lp^\ell$ coincides with the weighted term-over-position monomial order for $\F_{q^m}[x]$, since one could replace the $q$-degrees with normal degrees and get the classical cases.

We furthermore need the following definition in analogy to the weighted term-over-position monomial order:
\begin{defi}
The \emph{$(k_1,\dots,k_\ell)$-weighted $q$-degree} of  $[f_1(x) \;\dots \; f_\ell (x)]$ is defined as  $\max\{k_i+ \qdeg(f_i(x)) \mid i=1,\dots,\ell\}$.
\end{defi}

In the following we will not fix a monomial order. The results, if not noted differently, hold for any chosen monomial order.
\begin{defi}
We can order all monomials of an element $f\in\Lp^\ell$ in decreasing order with respect to some monomial order. Rename them such that $x^{[i_1]}e_{j_1}> x^{[i_2]}e_{j_2}> \dots $. Then
\begin{enumerate}
\item the \emph{leading monomial} $\mathrm{lm}(f)=x^{[i_1]} e_{j_1}$ is the greatest monomial of $f$.
\item the \emph{leading position} $\mathrm{lpos}(f)={j_1}$ is the vector coordinate of the leading monomial.
\item the \emph{leading term} $\mathrm{lt}(f)=f_{j_1, i_1}x^{[i_1]}e_{j_1}$ is the complete term of the leading monomial.
\end{enumerate}
\end{defi}

\vspace{0.2cm}

In order to define minimality for submodule bases we need the following notion of reduction, in analogy to \cite[Definition 4.1.1]{ad94b}. 

\begin{defi}
Let $ f, h \in \Lp^\ell$ and let $F=\{f^{(1)},\dots,f^{(s)}\}$ be a set of non-zero elements of $\Lp^\ell$. 
We say that \emph{ $f$ reduces to $h$ modulo $F$ in one step} if and only if
\[h= f- (( b_{1}x^{[a_{1}]}) \circ f^{(1)}+ \dots + (b_{k}x^{[a_{k}]} ) \circ f^{(k)} )\]
for some $a_{1},\dots,a_{k}\in \N_0$ and $b_{1},\dots, b_{k}\in \F_{q^m}$, where
\[\mathrm{lm}(f) = x^{[a_{i}]}\circ \mathrm{lm}(f^{(i)}), \quad i=1,\dots,k , \quad \textnormal{ and }\]
\[\mathrm{lt}(f) =( b_{1}x^{[a_{1}]})\circ \mathrm{lt}(f^{(1)})+ \dots + (b_{k}x^{a_{[k]}})\circ\mathrm{lt}(f^{(k)}).\]
We say that $f$ is \emph{minimal} with respect to $F$ if it cannot be reduced modulo $F$.
\end{defi}

\begin{defi}
A module basis $B$ is called \emph{minimal} if all its elements $b$ are minimal with respect to $B\backslash \{b\}$.
\end{defi}

\begin{prop}\label{prop:lpos}
Let $B$ be a basis of a module $M\subseteq\Lp^\ell$. Then $B$ is a minimal basis if and only if all leading positions of the elements of $B$ are distinct.
\end{prop}
\begin{proof}
Let $B$ be minimal. If two elements of $B$ have the same leading position, the one with the greater leading monomial can be reduced modulo the other element, which contradicts the minimality. Hence, no two elements of a minimal basis can have the same leading position.

The other direction follows straight from the definition of reducibility and minimality of a basis, since if the leading positions of all elements are different, none of them can be reduced modulo the other elements.
\end{proof}

The property outlined in the following theorem is well-established for minimal Gr\"obner bases for modules in $\F_q[x]^\ell$ with respect to multiplication. It extends to non-commutative Gr\"obner bases of solvable type, see e.g.~\cite[Lemma 1.5]{ka90}. As a result, it also holds over the ring of linearized polynomials. It was labeled \emph{Predictable Leading Monomial (PLM) property} in~\cite{ku11} to emphasize its closeness to Forney's \emph{Predictable Degree property}~\cite{fo75}. It captures the exact property that is needed in subsequent proofs. 

Note that in \cite{ku11} minimal bases were addressed as minimal Gr\"obner bases. It can be shown that in their setting as well as in the one of this paper a minimal basis is the same as a minimal Gr\"obner basis. For explicit formulation of the theory of  Gr\"obner bases for modules in $\Lp^\ell$ the interested reader is referred to our preprint paper~\cite{ku14arxiv}.

\begin{thm}[PLM property]\label{thm:PLM}
Let $M$ be a module in $\Lp^\ell$ with minimal basis $B=\{b^{(1)},\dots,b^{(L)}\}$. Then for any $0\neq f \in M$, written as
\[f=a_1(x)\circ b^{(1)}+\dots + a_L(x)\circ b^{(L)} ,\]
where $a_1(x),\dots,a_L(x)\in \Lp$, we have
\[\mathrm{lm}(f) = \max_{1\leq i \leq L; a_i(x)\neq 0} \{\mathrm{lm}(a_i)\circ \mathrm{lm}(b^{(i)})\}  \]
where $\mathrm{lm}(a_i(x))$ is the term of $a_i(x)$ of highest $q$-degree.
\end{thm}
\begin{proof}
Since $B$ is minimal, all leading positions and thus also all leading monomials of its elements are distinct (by Proposition \ref{prop:lpos}). 
Without loss of generality assume that $\mathrm{lm}(b^{(1)})> \mathrm{lm}(b^{(2)})> \dots > \mathrm{lm}(b^{(L)})$ and that all $a_i(x)$ are non-zero. Since $\Lp$ contains no zero divisors, we have that $\mathrm{lpos}(a_i(x)\circ b^{(i)})= \mathrm{lpos}(b^{(i)})$ for $1\leq i\leq L$. As a result, all leading positions and therefore all leading monomials of the $a_i(x)\circ b^{(i)}$'s are distinct. Thus there exist $j_1,\dots,j_L$ such that
\[\mathrm{lm}(a_{j_1}(x)\circ b^{(j_1)}) >\mathrm{lm}(a_{j_2}(x)\circ b^{(j_2)})>\dots >\mathrm{lm}(a_{j_L}(x)\circ b^{(j_L)}) .\]
It follows that 
\[\mathrm{lm}(f) =\mathrm{lm}(a_{j_1}(x)\circ b^{(j_1)})=\mathrm{lm}(a_{j_1}(x))\circ \mathrm{lm}(b^{(j_1)}) = \max_{1\leq i \leq L} \{\mathrm{lm}(a_i(x))\circ \mathrm{lm}(b^{(j_i)})\}  .\]
\end{proof}


\begin{prop}\label{prop:samelp}
The leading positions and weighted $q$-degrees of all elements of two distinct minimal bases for the same module in $\Lp^{\ell}$ have to be the same. This implies that the cardinality of both bases are equal as well.
\end{prop}
\begin{proof}
Let $B_1=\{ b^{(i)} \mid i=1,\dots, L\}$ and $B_2=\{ c^{(i)} \mid i=1,\dots, L'\}$ be two different minimal bases of the same module in $\Lp^\ell$. Then $ b^{(j)} $ must be a linear combination of the $c^{(i)}$ for $j=1,\dots,L$. Similarly, $ c^{(i)} $ must be a linear combination of the $b^{(j)}$ for $i=1,\dots,L'$. Hence, by the PLM property and since all leading positions are different in the bases, there exist $j' \in \{1,\dots,L'\}$ and $a(x),a'(x)\in \Lp$ such that $\lm(a(x)\circ c^{(j')}) =\lm (b^{(j)})$ and $\lm(a'(x)\circ b^{(j)}) =\lm (c^{(j')})$. 
This implies on the one hand that $\lpos(b^{(j)})=\lpos(c^{(j')})$ and on the other that $\qdeg(a(x))=\qdeg( a'(x))=0$, which implies that $\qdeg(b^{(j)})=\qdeg(c^{(j')})$.
\end{proof}


\section{Minimal List-Decoding of Gabidulin Codes}\label{sec:decoding}

In this section we describe a minimal list-decoding algorithm for Gabidulin codes. To explain this terminology further: given a received word, our list decoder outputs a list of exactly those message words that correspond to all codewords that are closest to the received word. The algorithm uses a parametrization within the interpolation module that is associated with the given received word. For this, we will need a minimal basis of this module, where minimality is with respect to the $(0,k-1)$-weighted $q$-degree. The construction of such a minimal basis will be described in the second subsection.

For the remainder of the paper 
let $g_1,\dots, g_n \in\F_{q^m}$ be linearly independent over $\F_q$ and let $M_k(g_1,\dots,g_n)$ be the generator matrix of the Gabidulin code $C\subseteq \F_{q^m}^n$. Let  $\mathbf{r}=(r_1,\dots,r_n) \in \F_{q^m}^n$ be the received word and denote $\mathbf{g}=(g_1,\dots,g_n) $.  
Moreover, throughout the remainder of this paper our monomial order will be the $(0,k-1)$-weighted term-over-position monomial order.


\subsection{The Parametrization}

In the following we abbreviate the row span of a (polynomial) matrix $A$ by $\rs(A)$.

\begin{defi}
The \emph{interpolation module} $\mathfrak{M}(\bf r)$ for $\mathbf r$ is defined as the left submodule of $\Lp^2$, given by
\[
\mathfrak{M}(\bf r) := \rs \left[\begin{array}{cc}  \Pi_\mathbf{g} (x) & 0 \\ -\Lambda_{\bf g,r}(x) & x \end{array}\right].
\]
\end{defi}

We identify any $[f(x) \quad g(x)] \in  \mathfrak{M}(\bf r)$ with the bivariate linearized $q$-polynomial $Q(x,y)= f(x) + g(y)$.
The following theorem shows that the name interpolation module is justified for $\mathfrak{M}(\bf r)$:

\begin{thm}\label{thm5}
$\mathfrak{M}(\bf r)$ consists exactly of all $Q(x,y)= f(x) +g(y)$ with $f(x), g(x) \in \Lp$, such that $Q(g_i,r_i)=0$ for $i=1,\dots,n$.
\end{thm}

\begin{proof}
For the first direction let $Q(x,y)= f(x) +g(y)$ be an element of $\mathfrak{M}(\bf r)$. Then there exist $\beta(x), \gamma(x) \in \Lp$ such that $f(x) = \beta(x)\circ\Pi_\mathbf{g}(x) - \gamma(x)\circ\Lambda_{\bf g,r}(x)$ and $\gamma(x) = g(x)$, thus $Q(g_i,r_i)= \beta(\Pi_\mathbf{g}(g_i)) - \gamma(\Lambda_{\bf g,r}(g_i)) +\gamma(r_i) = 0 -\gamma(r_i) +\gamma(r_i) = 0$.

For the other direction let $f(x), g(x) \in \Lp$ be such that $Q(g_i,r_i)= f(g_i) +g(r_i)=0$ for $i=1,\dots,n$. To show that $Q(x,y)\in\mathfrak{M}(\bf r)$ we need to find $\beta(x) \in \Lp$ such that
\[\beta(x)\circ \Pi_\mathbf{g}(x) - \gamma(x)\circ \Lambda_{\bf g,r}(x) = f(x) \quad\textnormal{ and }\quad \gamma(x) = g(x) .\]
We substitute the second into the first equation to get 
\begin{align}
\beta(x)\circ \Pi_\mathbf{g}(x)  = f(x)+ g(x)\circ \Lambda_{\bf g,r}(x) .
\end{align}
By assumption, the equation $f(g_i)+g( \Lambda_{\bf g,r}(g_i) )= f(g_i)+g(r_i) = 0$ holds for all $i$. Then, by Lemma \ref{lem3}, it follows that $f(x)+ g(x)\circ \Lambda_{\bf g,r}(x)$ is symbolically divisible on the right by $\Pi_\mathbf{g}(x)$ and hence there exists $\beta(x)\in \Lp$ such that $(1)$ holds.
\end{proof}

Combining all the previous results we get a description of all codewords with distance $t$ to the received word in the new parametrization:

\begin{thm}\label{thm:main}
The elements $f=[N(x) \quad -D(x)]$ of $\mathfrak{M}(\bf r)$  that fulfill
\begin{enumerate}
\item $\qdeg(N(x))\leq t+k-1$,
\item $\qdeg(D(x))=t$,
\item $N(x)$ is symbolically divisible on the left by $D(x)$, i.e.\ there exists $m(x)\in\Lp$ such that $D(m(x))=N(x)$,
\end{enumerate}
are in one-to-one correspondence with the codewords of rank distance $t$ to the received word $\mathbf r$. 

Note that conditions 1) and 2) are equivalent to that the $(0,k-1)$-weighted $q$-degree of $f$ is equal to $t+k-1$ and $\lpos(f)=2$.
\end{thm}
\begin{proof}
To prove the first direction let $\mathbf c \in \F_{q^m}^n$ be a codeword such that $d_R(\mathbf c, \mathbf r)=t$ with the corresponding message polynomial $m(x)\in\Lp_{<k}$. Then by Theorem \ref{thm2} there exists $D(x)\in\Lp$ of $q$-degree $t$ such that $D(m(g_i))=D(r_i)$ for $i=1,\dots,n$. By Theorem \ref{thm5} we know that $[D(m(x)) \quad -D(x)]$ is in $ \mathfrak{M}(\bf r)$. It holds that $\qdeg(D(m(x)))\leq t+k-1$ and that $(D(m(x))$ is symbolically divisible on the left by $D(x)$.

For the other direction let $[N(x) \quad -D(x)] \in \mathfrak{M}(\bf r)$ fulfill conditions $1)-3)$. Then we know that the divisor $m(x)\in \Lp$ has $q$-degree less than $k$ and that $N(x)=D(m(x))$. Since it is in $ \mathfrak{M}(\bf r)$ we know by Theorem \ref{thm5} that $D(m(g_i))-D(r_i)=0$ for all $i$. Define $\mathbf{ c}   :=(m(g_1), \dots, m(g_n))$, then it follows from Theorem \ref{thm2} that $d_R(\mathbf c, \mathbf r)=t$.
\end{proof}

Therefore, list decoding within rank radius $t$ is equivalent to finding all elements  $[N(x) \quad -D(x)]$ in $\mathfrak{M}(\bf r)$ with $(0,k-1)$-weighted $q$-degree less than $t+k$ and leading position $2$, 
such that $N(x)$ is symbolically divisible on the left by $D(x)$.
%
Note that this is a generalization of the interpolation-based decoding method from \cite{lo06}. The difference is that our method can also decode beyond the unique decoding radius.

We can now describe the list decoding algorithm. Since in most applications one wants to find the set of all closest codewords to the received word, our algorithm will do exactly this. In contrast, a complete list decoder with a prescribed radius $t$ finds all codewords within radius $t$ from the received word, even if some of them are closer than others.


Our decoding algorithm is Algorithm \ref{alg1} below. This algorithm requires the computation of a minimal basis of $\mathfrak{M}(\bf r)$ (only once for each received word). The algorithm then iteratively searches for all elements in $\mathfrak{M}(\bf r)$ of $(0,k-1)$-weighted $q$-degree $t+k-1$ for increasing $t$ and checks the requirements of Theorem \ref{thm:main}. As soon as solutions are found, $t$ will not be increased any further and the algorithm terminates.

As in Section \ref{sec:modules} we use the notation $f = [f_{1}(x) \;\;\; f_2(x)], b^{(i)} = [b^{(i)}_{1}(x) \;\;\; b^{(i)}_2(x)]$, etc.\ for elements of the interpolation module  $\mathfrak{M}(\bf r)$.

\begin{algorithm}
\caption{Minimal list decoding of a Gabidulin code  of dimension $k$ with generators $g_1,\dots, g_n \in \F_{q^m}$.}
\label{alg1}
\begin{algorithmic}
\REQUIRE Positive integers $k, n$; ${\bf g} =(g_1 , \ldots , g_n ) \in \F_{q^m}^n$, received word ${\bf r} \in \F_{q^m}^n$. 
\STATE 1. Compute $\Pi_{\mathbf g} (x)$ and $\Lambda_{\bf g,r}(x)$, both in $\Lp$. Define the interpolation module $$\mathfrak{M}(\bf r) := \rs \left[\begin{array}{cc}  \Pi_{\mathbf g}(x) & 0 \\ -\Lambda_{\bf g,r}(x) & x \end{array}\right]  .$$
\STATE 2. Compute a minimal basis $B=\{b^{(1)}, b^{(2)}\}$ of $\mathfrak{M}(\bf r) $ with respect to the $(0,k-1)$-weighted degree, with 
$\lpos(b^{(1)}) =1$ and $\lpos(b^{(2)}) =2$.
\STATE 3. Define $\ell_1, \ell_2$ as the  $(0,k-1)$-weighted $q$-degrees of $b^{(1)}, b^{(2)}$, respectively.
\STATE 4. Define \texttt{list}$:=[ \; ]$ (an empty list) and $j:=0$.
\WHILE{\texttt{list}$=[\; ]$}
\FORALL{$\beta(x)\in \Lp, \qdeg(\beta(x))\leq \ell_2-\ell_1+j$}
\FORALL{monic $\gamma(x) \in \Lp, \qdeg(\gamma(x))=j$}
\STATE $f := \beta(x)\circ b^{(1)} + \gamma(x)\circ b^{(2)}$
\IF{there exists $ m(x) \in \Lp$ such that $f_1(x) = f_2(m(x))$} 
\STATE add $m(x)$ to \texttt{list}
\ENDIF
\ENDFOR
\ENDFOR
\STATE $j:=j+1$
\ENDWHILE
\RETURN \texttt{list}
\end{algorithmic}
\end{algorithm}

\begin{thm}
Algorithm \ref{alg1} yields a list of all message polynomials such that the corresponding codewords are closest to the received word.
\end{thm}
\begin{proof}
We will prove this in two steps; first  we show that any closest codeword will be in the output list, then we show that any element in the output list of the algorithm is a closest codeword.

For the first direction let $\mathbf c\in \F_{q^m}^n$ be a closest codeword with corresponding message polynomial $m(x) \in \Lp_{<k}$ and $t:=d_R(\mathbf c,\mathbf r)$. 
We know from Theorem \ref{thm:main} that there exists $f= [f_{1}(x) \;\;\; f_2(x)] = [f_2(m(x)) \;\;\; f_2(x)] \in \mathfrak{M}(\bf r)$, with $\lpos(f)=2$ and $\qdeg(f_2(x)) =t$. I.e.\ there exist $\beta(x), \gamma(x) \in \Lp$ such that $\beta(x) \circ b^{(1)} + \gamma(x) \circ b^{(2)} = f$. 
Furthermore, we know from the PLM property (Theorem \ref{thm:PLM}), that 
\[\lm(f) = \max \{\lm(\beta(x))\circ \lm(b^{(1)}) , \lm(\gamma(x))\circ \lm(b^{(2)})\}\]
which implies that $\lm(f) = \lm(\gamma(x))\circ \lm(b^{(2)}) $, since $\lpos(b^{(i)})=i$ for $i=1,2$. Therefore we get
\[  t= \qdeg(\gamma(x)) +(\ell_2-k+1) \iff  \qdeg(\gamma(x)) = t-(\ell_2-k+1)  .\]
Because of the leading position we furthermore get
\[  \qdeg(\beta(x)) + \qdeg(b_1^{(1)}(x)) \leq \qdeg(\gamma(x)) + \qdeg(b_2^{(2)}(x)) + k-1\]
\[\iff \qdeg(\beta(x))  \leq \qdeg(\gamma(x)) + \ell_2 - \ell_1  .\]
Set $j:= t-\ell_2 +k -1$, then the above conditions translate into $\qdeg(\gamma(x))=j$ and $\qdeg(\beta(x)) \leq j+\ell_2- \ell_1$, which is exactly the parametrization used in the algorithm.
We can choose $\gamma(x)$ monic, since the non-zero scalar coefficients of our polynomials are from $\F_{q^m}$ and are thus invertible. 

Note that increasing $j$ by one is equivalent to increasing $t$ by one. Therefore, if $t$ is the minimal distance between any codeword and the received word, the algorithm will not terminate at any $j<t-\ell_2 +k -1$. Thus, we have shown that the algorithm will produce $m(x)$ as an output element.

For the other direction let $m(x)$ be an element of the output list of Algorithm \ref{alg1}. Then there exists $f= [f_{1}(x) \;\;\; f_2(x)] = [f_2(m(x)) \;\;\; f_2(x)] \in \mathfrak{M}({\bf r})$. Furthermore there exist $j\in \N_0$ and $\beta(x), \gamma(x) \in \Lp$ with $\qdeg(\beta(x))\leq \ell_2 - \ell_1 +j, \qdeg(\gamma(x))=j$ such that $\beta(x) \circ b^{(1)} + \gamma(x) \circ b^{(2)} = f$. 
Since $\qdeg(b_1^{(1)})= \ell_1, \qdeg(b_2^{(1)})< \ell_1-k+1$ and $\qdeg(b_1^{(2)})\leq \ell_2, \qdeg(b_2^{(2)})< \ell_2-k+1$, it follows that
\[\qdeg(f_{1}(x)) \leq \ell_2 + j \quad, \quad \qdeg(f_2(x)) = \ell_2 + j-k +1 .  \]
Hence all the requirements of Theorem \ref{thm:main} are fulfilled for $t=\ell_2 + j-k +1$ and $m(x)$ belongs to a codeword $\mathbf{c} = (m(g_1), \dots, m(g_n))$ with $d_R(\mathbf{ c}, \mathbf r) =\ell_2 + j-k +1$.
It remains to show that there is no closer codeword. Assume there exists $m'(x) \in \Lp$ of $q$-degree less than $k$ with corresponding codeword $\mathbf c'$ such that $t':=d_R(\mathbf c', r) < \ell_2 + j-k +1$. Then, by the first part of this proof, the algorithm would terminate at $j=t'-\ell_2 +k -1 < t-\ell_2 +k -1$ and would not produce $m(x)$ as an output element.

It remains to show that there are no codewords at rank distance less than $\ell_2-k+1$, since this is the distance for the initial loop with $j=0$. 
Assume there would be such a codeword with corresponding message polynomial $m(x)\in\Lp$. Then there exists $D(x)\in\Lp$ with $q$-degree less than $\ell_2-k+1$ such that $f' := [D(m(x)) \;\; D(x)]$ is in $\mathfrak{M}(\bf r)$. Then the $(0,k-1)$-weighted $q$-degree of $f'$ is less than $\ell_2$, which means that $B$ is not a minimal basis of $\mathfrak{M}(\bf r)$, which is a contradiction.
\end{proof}


\subsection{Construction of a Minimal Basis}\label{subsec:basisconstruction}

We will now explain two different ways of obtaining the minimal basis for the interpolation module, as required in Algorithm \ref{alg1}. The first one is the extended Euclidean algorithm for composition while the second one is an iterative algorithm. 
Note that unique decoding with the Euclidean algorithm was done in a Gao-like algorithm in \cite{wa13phd}. Similarly, unique decoding (possibly beyond half the minimum distance as a list-$1$ decoder) with an iterative algorithm analogous to our Algorithm \ref{alg3} can be found in \cite{ko08,lo06,xi11}. Our algorithm differs from these works in the sense that we compute all (and not only one) closest codewords to the received word. 
For this, our set-up in terms of modules and particularly our parametrization result are novel ingredients that we believe give new insights into this topic.

\vspace{0.5cm}

The next algorithm for constructing a minimal basis of $\mathfrak{M}(\mathbf r)$ is the extended Euclidean algorithm (EEA) for $q$-linearized polynomials with respect to composition. This variant of the Euclidean algorithm is well-known, see e.g.\ \cite[Algorithm 2.3]{wa13phd}, and works analogously to the classical EEA for normal polynomials. Since we need to distinguish left and right symbolic division, there exists a right and left EEA in $\Lp$. In this work we are only interested in the right EEA. The main ingredient is right symbolic division with remainder: For given $f(x), g(x) \in \Lp$ with $\qdeg(f(x)) \geq \qdeg (g(x))$ this type of division computes $q(x), r(x) \in \Lp$ such that $\qdeg (r(x)) < \qdeg (g(x))$ and 
\[f(x) = q(x) \circ g(x) + r(x) .\]

\begin{algorithm}[hb]
\caption{Computation of minimal basis of $\mathfrak{M}(\mathbf r)$ via the extended Euclidean algorithm.}
\label{alg2}
\begin{algorithmic}
\REQUIRE Positive integers $k,n$; ${\mathbf g} \in \F_{q^m}^n$, received word ${\bf r} =(r_1 ,\ldots , r_n) \in \F_{q^m}^n$; $q$-annihilator polynomial $\Pi_{\mathbf g}(x)$ and $q$-Lagrange polynomial $\Lambda_{\bf g,r}(x)$.
\STATE Initialize $j=0$ and define $P_0(x),N_0(x),K_0(x), D_0(x)$ as 
\[ \left[\begin{array}{cc}P_0(x) & K_0(x) \\ N_0(x) & D_0(x) \end{array}\right]  :=  \left[\begin{array}{cc}  \Pi_{\mathbf g}(x) & 0 \\ -\Lambda_{\bf g,r}(x) & x \end{array}\right]  . \]
	\WHILE{$\qdeg(D_{j}) + k-1 < \qdeg(N_{j})$}
\STATE Apply right symbolic division to compute $q_{j}(x), r_{j}(x) \in \Lp$ such that $P_j(x) = q_{j}(x) \circ N_{j}(x) + r_{j}(x)$ and $\qdeg(r_{j})< \qdeg(N_{j})$.
\STATE Update 
\[ \left[\begin{array}{cc}P_{j+1}(x) & K_{j+1}(x) \\ N_{j+1}(x) & D_{j+1}(x) \end{array}\right]  :=  \left[\begin{array}{cc}  N_{j}(x) & D_{j}(x) \\  r_j(x) &  K_j (x) - q_{j}(x) \circ D_{j}(x) \end{array}\right]  . \]
\STATE Set $j:= j+1$.
\ENDWHILE
\RETURN $b^{(1)} := [\: P_j(x) \quad K_j(x)\:]$ and $b^{(2)} :=  [\: N_{j}(x) \quad D_{j}(x)\:]$
\end{algorithmic}
\end{algorithm}

\begin{thm}
Algorithm \ref{alg2} produces a minimal basis $B=\{b^{(1)}, b^{(2)}\}$ for our interpolation module $\mathfrak{M}(\bf r)$. Moreover, $\lpos(b^{(1)})=1$ and $\lpos(b^{(2)})=2$.
\end{thm}
\begin{proof}
As shown in Theorem \ref{thm5}, for $j=0$ we have a basis of $\mathfrak{M}(\bf r)$. Then for $j\geq 1$ we have that
\[ \left[\begin{array}{cc}P_{j}(x) & K_{j}(x) \\ N_{j}(x) & D_{j}(x) \end{array}\right]  =  \left[\begin{array}{cc}  0 & x \\  x &  - q_{j-1}(x) \end{array}\right]  \circ  \left[\begin{array}{cc}P_{j-1}(x) & K_{j-1}(x) \\ N_{j-1}(x) & D_{j-1}(x) \end{array}\right]    \]
\[\iff   \left[\begin{array}{cc}  q_{j-1}(x) & x \\  x &  0 \end{array}\right]  \circ  \left[\begin{array}{cc}P_{j}(x) & K_{j}(x) \\ N_{j}(x) & D_{j}(x) \end{array}\right]  = \left[\begin{array}{cc}P_{j-1}(x) & K_{j-1}(x) \\ N_{j-1}(x) & D_{j-1}(x) \end{array}\right]   , \]
hence also the new matrix is a basis of $\mathfrak{M}(\bf r)$.

We know that at level $0$ both basis vectors have leading position $1$. The \textbf{while}-condition assures that the algorithm terminates as soon as we have a basis where the second vector has leading position $2$, while the first still has leading position $1$.  This automatically implies that this basis is minimal, according to Proposition \ref{prop:lpos}.
\end{proof}

\begin{ex}\label{ex15}
Consider the Gabidulin code in $\F_{2^3}\cong \F_2[\alpha]$ (with $\alpha^3=\alpha +1$) of dimension $k=2$ with generator matrix
\[
G= \left( \begin{array}{ccc} 1 & \alpha & \alpha^2 \\ 1 & \alpha^2 & \alpha^4\end{array}\right) .
\]
Thus $\mathbf g = (g_1 , g_2, g_3 )= (1,\alpha, \alpha^2)$. Suppose that the received word equals
\[\mathbf{ r} =(\: \alpha +1 \;0 \; \alpha \:) .\]
Then we construct the interpolation module
$$M(\bf r) = \rs \left[\begin{array}{cc}  \Pi_{\mathbf g}(x) & 0 \\ -\Lambda_{\bf g, \bf r}(x) & x \end{array}\right]  = \rs \left[\begin{array}{cc}  x^8 + x & 0 \\  \alpha^2 x^4+ \alpha^5 x & x \end{array}\right] .$$
To compute a minimal basis we apply symbolic division and get
\[x^8+x  = (\alpha^3 x^2) \circ ( \alpha^2 x^4+ \alpha^5 x) + \alpha^6 x^2 +  x .\]
Since $\qdeg( \alpha^3 x^2) +k-1 =2 \geq 1=\qdeg( \alpha^6 x^2 +  x)$, the algorithm terminates and a minimal basis (w.r.t. the $(0,1)$-weighted $2$-degree) of this module is
\[ \left[\begin{array}{cc}g_1^{(1)} & g_1^{(2)}\\ g_2^{(1)} &g_2^{(2)}\end{array}\right]  =  \left[\begin{array}{cc} \alpha^2 x^4+ \alpha^5 x & x\\ \alpha^6 x^2 +  x  & \alpha^3 x^2 \end{array}\right] . \]
Hence we get $\ell_1 =2$ and $ \ell_2= 2$, i.e.\ we want to use all $\beta(x)\in\mathcal{L}_2(x,2^3)$ with $2$-degree less than or equal to $0$ and all monic $\gamma(x)\in\mathcal{L}_2(x,2^3)$  with $2$-degree equal to $0$. Thus, $\beta(x)= b_0 x$ for $b_0\in\F_{2^3}$ and $\gamma(x)= x$. We get divisibility for $b_0\in\F_{2^3}\backslash \{0\}$. The corresponding message polynomials and codewords are
$$\begin{array}{ll}
m_1(x) =  {x^2 + \alpha x}  &\quad  c_1=(\: \alpha+1 \quad 0 \quad \alpha^2+1),\\
m_2(x) = \alpha^5 x^2 + \alpha^2 x  &\quad c_2=(\: \alpha+1  \quad \alpha \quad \alpha),\\
m_3(x) = \alpha^3 x^2 + \alpha^4 x  &\quad c_3=(\: \alpha^2 +1 \quad 0 \quad \alpha^2),\\
m_4(x) = \alpha^4 x^2   &\quad c_4=(\: \alpha^2 +\alpha \quad \alpha^2+1 \quad \alpha),\\
m_5(x) = \alpha^6 x^2 + \alpha^6 x  &\quad c_5=(\: 0 \quad \alpha+1 \quad 1),\\
m_6(x) = \alpha^2 x^2 + \alpha^3 x  &\quad c_6=(\: \alpha^2+\alpha+1 \quad 0 \quad \alpha),\\
m_7(x) = \alpha x^2 +  x  &\quad c_7=(\: \alpha+1 \quad 1 \quad \alpha+1).
\end{array}$$
All these codewords are rank distance $1$ away from $\mathbf{ r} =(\: \alpha +1 \;0 \; \alpha \:)$. Note that their Hamming distance to $\bf r$ varies from $1$ to $3$.
\end{ex}

\vspace{0.5cm}


We will now derive an alternative, more efficient algorithm, namely an iterative algorithm for the computation of a minimal basis of the interpolation module. For this we first need the following result:

\begin{lem}\label{lem:it_min_bas}
For $i=1,\dots,n$ denote by $\mathfrak{M}_i$ the interpolation module for $(g_1,\dots,g_i)$ and $(r_1,\dots,r_i)$. Let 
\[\left[\begin{array}{cc}  P(x) & -K(x) \\ N(x) & -D(x) \end{array}\right]\]
be a basis for $\mathfrak{M}_{i-1}$ and
$$\Gamma_i := P(g_i)-K(r_i)  \quad, \quad \Delta_i := N(g_i)- D(r_i) .$$ 
If $\Gamma_i\neq 0$, then the row vectors of 
\[\left[\begin{array}{cc}  b^{(1)} \\ b^{(2)} \end{array}\right] := \left[\begin{array}{cc}  x^q-  \Gamma_i^{q-1} x & 0 \\ \Delta_i x& -\Gamma_i x\end{array}\right] \circ \left[\begin{array}{cc}  P(x) & -K(x) \\ N(x) & -D(x) \end{array}\right]\]
form a basis of $\mathfrak{M}_i$.
If $\Delta_i\neq 0$, then the row vectors of 
\[\left[\begin{array}{cc}   \Delta_i x& -\Gamma_i x \\ 0& x^q-  \Delta_i^{q-1} x  \end{array}\right] \circ \left[\begin{array}{cc}  P(x) & -K(x) \\ N(x) & -D(x) \end{array}\right]\]
form a basis of $\mathfrak{M}_i$.
\end{lem}
\begin{proof}
We first consider the first case and show that both $b^{(1)}$ and $b^{(2)}$ are in $\mathfrak{M}_i$. From the assumptions it follows that $P(g_j)=K(r_j)$ and that $N(g_j)=D(r_j)$ for $1\leq j<i$. Moreover, the two entries of $b^{(1)}$ are given by
$$(x^q- \Gamma_i^{q-1} x) \circ P(x) = P(x)^q-\Gamma_i^{q-1} P(x), $$
$$(x^q- \Gamma_i^{q-1} x) \circ K(x) = K(x)^q- \Gamma_i^{q-1} K(x) ,$$
thus
$P(g_j)^q- \Gamma_i^{q-1} P(g_j) - K(r_j)^q+ \Gamma_i^{q-1} K(r_j)=0$ for $1\leq j\leq i$.
For $b^{(2)}$ we get
$$\Delta_i P(g_j) - \Gamma_i N(g_j) - \Delta_i K(r_j) + \Gamma_i D(r_j) =$$ $$ \Delta_i (P(g_j) -K(r_j)) - \Gamma_i (N(g_j) -D(r_j)) = \Delta_i \Gamma_i - \Gamma_i \Delta_i= 0$$
for $1\leq j\leq i$. Thus, $b^{(1)}$ and $b^{(2)}$ are elements of $\mathfrak{M}_i$. 

It remains to show that $b^{(1)}$ and $b^{(2)}$ span the entire interpolation module (and not just a submodule of it). For this, it is sufficient to show that $[\; \Pi_{i-1}(x) \quad 0 \;]$ and $[\; \Lambda_{i-1}(x) \quad -x \;]$ are linear combinations of  $b^{(1)}$ and $b^{(2)}$. 
Since $[\; P(x) \quad -K(x) \;]$ and $[\; N(x) \quad -D(x) \;] $ form a basis of $\mathfrak{M}_i$, there exist $\bar \beta(x), \bar \gamma(x) \in \Lp$ such that 
$$\bar \beta(x)\circ [\; P(x) \quad -K(x) \;] + \bar \gamma(x) \circ [\; N(x) \quad -D(x) \;] = [\; \Pi_{i-1}(x) \quad 0 \;].$$ 
Let  $\beta(x), \gamma(x)\in\Lp$ be such that
$$\beta(x)\circ (x^q-\Gamma_i^{q-1}x) =(x^q-\Pi_{i-1}(g_i)^{q-1} x)\circ\left( \bar \gamma(\frac{\Delta_i}{\Gamma_i}x)+\bar \beta(x)\right),$$ 
$$\gamma(x) = -(x^q-\Pi_{i-1}(g_i)^{q-1} x)\circ \bar \gamma(\frac{1}{\Gamma_i}x)  .$$ 
Note that it can easily be checked that $\Gamma_i$ is a root of the right side of the previous equation, thus $\beta(x)$ is well-defined by Lemma \ref{lem3} . 
Denote the first and second row of the new basis by $b^{(1)}$ and $b^{(2)}$, respectively. Then 
$$\beta(x)\circ b^{(1)}_1 + \gamma(x) \circ b^{(2)}_1 = \beta(x) \circ (x^q- \Gamma_i^{q-1} x) \circ  P(x)    +  \gamma(x)\circ ( \Delta_i P(x)  -\Gamma_i  N(x) ) $$
$$ = (x^q-\Pi_{i-1}(g_i)^{q-1} x)\circ\left( \bar \gamma(\frac{\Delta_i}{\Gamma_i}x)+\bar \beta(x)\right) \circ  P(x) +      \gamma(\Delta_i P(x)) - \gamma(\Gamma_i N(x))  $$
$$ = \left( (x^q-\Pi_{i-1}(g_i)^{q-1} x)\circ \bar  \beta(x) - \gamma(\Delta_i x) \right) \circ  P(x) +      \gamma(\Delta_i P(x)) - \gamma(\Gamma_i N(x))  $$
$$ = \left( (x^q-\Pi_{i-1}(g_i)^{q-1} x)\circ \bar  \beta(x) \right) \circ  P(x) +  (x^q-\Pi_{i-1}(g_i)^{q-1} x)\circ \bar \gamma( N(x))  $$
$$ = (x^q-\Pi_{i-1}(g_i)^{q-1} x)\circ \left( \bar  \beta(  P(x)) +  \bar \gamma( N(x)) \right) $$ $$=  (x^q-\Pi_{i-1}(g_i)^{q-1} x)\circ \Pi_{i-1}(x) = \Pi_i(x)  , $$
and 
$$\beta(x)\circ b^{(1)}_2 + \gamma(x) \circ b^{(2)}_2 = -\beta(x) \circ (x^q- \Gamma_i^{q-1} x) \circ  K(x)    -  \gamma(x)\circ ( \Delta_i K(x)  -\Gamma_i  D(x) ) $$
$$ = \left( (x^q-\Pi_{i-1}(g_i)^{q-1} x)\circ \bar  \beta(x) - \gamma(\Delta_i x) \right) \circ  K(x) +      \gamma(\Delta_i K(x)) - \gamma(\Gamma_i D(x))  $$
$$ = (x^q-\Pi_{i-1}(g_i)^{q-1} x)\circ \left( \bar  \beta(  K(x)) +  \bar \gamma( D(x)) \right) = 0 . $$
Thus 
$\beta(x)\circ b^{(1)} + \gamma(x) \circ b^{(2)} = [\: \Pi_{i}(x) \quad 0\: \;]  ,$
 i.e.\  $[\; \:\Pi_{i}(x) \quad 0\: \;]$ is in the module spanned by the new basis. 

Analogously, if we have that $ \bar c(x)\circ [\; P(x) \quad -K(x) \;] + \bar d(x) \circ [\; N(x) \quad -D(x) \;] = [\; \Lambda_{i-1}(x) \quad -x \;] $ and define $c(x), d(x) \in \Lp$ such that
$$d(x) =- \left(\bar d(x) + \frac{\Lambda_{i-1}(g_i) - r_i}{\Pi_{i-1}(g_i)} \bar \gamma(x)\right) \circ (\frac{1}{\Gamma_i} x) $$
$$c(x) \circ (x^q-\Gamma_i^{q-1}x) = \bar c(x) +\frac{\Lambda_{i-1}(g_i) - r_i}{\Pi_{i-1}(g_i)} \bar \beta(x) - d(\Delta_ix) $$
we get
$$c(x)\circ b^{(1)} + d(x) \circ b^{(2)} = [\; \Lambda_{i}(x) \quad -x \;] .$$
Hence, we have shown that the new basis $\{b^{(1)}, b^{(2)}\}$ spans the entire interpolation module.

For the second case note that
$$ \rs \left(\left[\begin{array}{cc}   \Delta_i x& -\Gamma_i x \\ 0& x^q-  \Delta_i^{q-1} x  \end{array}\right] \circ \left[\begin{array}{cc}  P(x) & -K(x) \\ N(x) & -D(x) \end{array}\right] \right)
$$
$$
= \rs \left(\left[\begin{array}{cc}    x^q-  \Delta_i^{q-1} x & 0 \\  \Gamma_i x & -\Delta_i x  \end{array}\right] \circ \left[\begin{array}{cc}  N(x) & -D(x)\\ P(x) & -K(x)  \end{array}\right] \right),$$
which corresponds to the first case after exchanging $P(x)$ with $N(x)$ and $K(x)$ with $D(x)$ (and vice versa).
\end{proof}

\begin{rem}
In the notation of Proposition~\ref{prop:Lagrec}, applying the previous theorem to $P(x)=\Pi_{i-1}(x), K(x)=0, N(x)=\Lambda_{i-1}(x)$ and $D(x)=-x$, leads to a computation that is identical up to a constant to the one in Proposition~\ref{prop:Lagrec} in which the $q$-annihilator polynomial and the $q$-Lagrange polynomial are iteratively constructed.
\end{rem}

Using Lemma \ref{lem:it_min_bas} as our main ingredient, we now set out to design an iterative algorithm that computes a minimal basis for $\mathfrak{M}_i$ at each step $i$.

\begin{algorithm}
\caption{Iterative computation of a minimal basis of $\mathfrak{M}(\mathbf r)$.}
\label{alg3}
\begin{algorithmic}
\REQUIRE Positive integers $k,n$; ${\mathbf g} =(g_1 ,\ldots , g_n) \in \F_{q^m}^n$, received word ${\bf r} =(r_1 ,\ldots , r_n) \in \F_{q^m}^n$.
\STATE Initialize   $j:=0$, $B_0 :=  \left[\begin{array}{cc}x& 0 \\ 0 & x \end{array}\right] $ .
\STATE We denote $B_i := \left[\begin{array}{cc}  P_i(x) & -K_i(x) \\ N_i(x) & -D_i(x) \end{array}\right] $.
\FOR{$i$ from $1$ to $n$}
\STATE  
$ \Gamma_i := P_{i-1}(g_i)-K_{i-1}(r_i)\quad, \quad  \Delta_i := N_{i-1}(g_i)- D_{i-1}(r_i).$
\IF{[$\qdeg(P_{i-1}(x))\leq \qdeg(D_{i-1}(x))+k-1\;  \AND \; \Gamma_i\neq 0$] or  $\Delta_i =0$}
\vspace{0.4cm}
\STATE
$B_i := \left[\begin{array}{cc}  x^q-\Gamma_i^{q-1} x & 0 \\ \Delta_i x& -\Gamma_i x\end{array}\right] \circ B_{i-1}$
\ELSE
\vspace{0.4cm}
\STATE
$B_i:=\left[\begin{array}{cc} \Delta_i x& -\Gamma_i x \\  0 & x^q-\Delta_i^{q-1} x \end{array}\right] \circ B_{i-1}$
\ENDIF
\ENDFOR
\RETURN \textbf{$B_n$}
\end{algorithmic}
\end{algorithm}
\begin{thm}
Algorithm \ref{alg3} yields a minimal basis of the interpolation module $\mathfrak{M}(\bf r )$, where the leading position of the first row is $1$ and the leading position of the second row is $2$.
\end{thm}
\begin{proof}
Denote by $M_1$ the matrix we multiply by on the left in the first IF statement and by $M_2$ the one in the ELSE statement of the algorithm. 
We know from Lemma \ref{lem:it_min_bas} that at each step, $B_i$ is a basis for the interpolation module $\mathfrak M_i$. 
We now show that it is a minimal basis with respect to the $(0,k-1)$-weighted term-over-position monomial order via induction on $i$. Assume that at step $i$ the first row has leading position $1$ and the second row has leading position $2$, i.e.\ $\qdeg(P_i(x)) > \qdeg(K_i(x))+k-1$ and $\qdeg(N_i(x)) \leq \qdeg(D_i(x))+k-1$.  
If $\qdeg(P_i(x))\leq \qdeg(D_i(x))+k-1$ we composite on the left by $M_1$. Hence, 
$$ \qdeg(P_{i+1}(x)) = \qdeg(P_i(x)) +1 $$
and
$$ \qdeg(K_{i+1}(x)) = \qdeg(K_i(x)) +1 < \qdeg(P_{i}(x)) -k+2 =   \qdeg(P_{i+1}(x)) -k+1.$$ 
Thus, the leading position of the first row of $B_{i+1}$ is still $1$. 
Moreover, 
$$ \qdeg(N_{i+1}(x)) \leq \max\{\qdeg(P_i(x)) , \qdeg(N_i(x))\} \leq  \qdeg(D_i(x))+k-1$$
and, since the assumptions imply that $\qdeg(K_i(x)) < \qdeg(D_i(x))$,
{\small $$ \qdeg(D_{i+1}(x)) = \max\{\qdeg(K_i(x)) , \qdeg(D_i(x))\} = \qdeg(D_i(x)) .$$}
Thus the leading position of the second row is $2$. Since the assumptions are true for $B_0$ the statement follows via induction. 

Analogously one can prove that composition with $M_2$ yields a basis of $\mathfrak M_i$ with different leading positions in the two rows. I.e.\ at each step we get a basis of $\mathfrak M_i$ with different leading positions, which is by Proposition \ref{prop:lpos} a minimal basis.  
Thus, after $n$ steps, $B_n$ is a minimal basis for the interpolation module $\mathfrak M(\bf r)$. 
\end{proof}

\begin{rem}
It can be verified that, due to the linear independence of $g_1, \ldots , g_k$, the first $k$ steps of the algorithm coincide up to a constant with the computation in Proposition~\ref{prop:Lagrec}. In other words, up to a constant, at step $k$ the algorithm has computed the $q$-annihilator polynomial and the $q$-Lagrange polynomial corresponding to the data so far.
\end{rem}

\begin{ex}
Consider the same setting as in Example \ref{ex15}, i.e.\ a
 Gabidulin code in $\F_{2^3}\cong \F_2[\alpha]$ (with $\alpha^3=\alpha +1$) with generator matrix 
\[G= \left( \begin{array}{ccc} 1 & \alpha & \alpha^2 \\ 1 & \alpha^2 & \alpha^4\end{array}\right) \]
and the received word $\mathbf{ r} =(\: \alpha^3 \;0 \; \alpha \:) $. We iteratively compute
$$B_1=\left[\begin{array}{cc}  x^2+x & 0 \\ (\alpha+1)x & x \end{array}\right] , $$
$$B_2=\left[\begin{array}{cc}  x^4+(\alpha^2 + \alpha+1)x^2+(\alpha^2 + \alpha)x & 0 \\ (\alpha^2+\alpha)x^2 +(\alpha^2 + \alpha+1)x & (\alpha^2+\alpha)x \end{array}\right] ,$$
$$B_3=\left[\begin{array}{cc}  \alpha^2 x^4+ \alpha^5 x & x\\ \alpha x^4 +\alpha^4 x^2 + x & \alpha x^2+\alpha^6 x \end{array}\right] .$$
$B_3$ is a minimal $(0,1)$-weighted basis of the interpolation module. 
We get $\ell_1 =2$ and $ \ell_2= 2$, i.e.\ we want to use all $\beta(x)\in\mathcal{L}_2(x,2^3)$ with $2$-degree less than or equal to $0$ and all monic $\gamma(x)\in\mathcal{L}_2(x,2^3)$  with $2$-degree equal to $0$. Thus, $\beta(x)= b_0 x$ for $b_0\in\F_{2^3}$ and $\gamma(x)= x$. We get divisibility for $b_0\in \F_{2^3}\backslash \{ \alpha^6\}$. The corresponding message polynomials are indeed the same as the ones from Example \ref{ex15}, although the minimal basis of the interpolation module differs from the one in Example \ref{ex15}. This can also be verified by the fact that
$$\left[\begin{array}{cc}  \alpha^2 x^4+ \alpha^5 x & x\\ \alpha^6 x^2 + x & \alpha^3 x^2 \end{array}\right] = \left[\begin{array}{cc}  x & 0 \\ \alpha x & \alpha^2 x \end{array}\right] \circ \left[\begin{array}{cc}  \alpha^2 x^4+ \alpha^5 x & x\\ \alpha x^4 +\alpha^4 x^2 + x & \alpha x^2+\alpha^6 x \end{array}\right] $$
$$\implies  [\; \beta'(x) \quad x \; ] \circ \left[\begin{array}{cc}  \alpha^2 x^4+ \alpha^5 x & x\\ \alpha^6 x^2 + x & \alpha^3 x^2 \end{array}\right] =$$ $$ \alpha^2 [\; \alpha^5 \beta'(x) + \alpha^6 x \quad x \; ] \circ  \left[\begin{array}{cc}  \alpha^2 x^4+ \alpha^5 x & x\\ \alpha x^4 +\alpha^4 x^2 + x & \alpha x^2+\alpha^6 x \end{array}\right] ,$$
which implies that $\beta'(x)$ in the setting of Example \ref{ex15} corresponds to $\beta(x) = \alpha^5 \beta'(x) + \alpha^6 x$ in this example.
\end{ex}

The next three lemmas present several properties of the two elements of a minimal basis of the interpolation module; these properties are later used to prove our main result.

\begin{lem}\label{lem:l1andl2}
Consider a Gabidulin code $C\subseteq  \F_{q^m}^n$  of dimension $k$.
Let $\mathfrak{M}(\bf r)$ be the interpolation module of the received word $\mathbf{ r}\in \F_{q^m}^n$ with minimal basis $B=\{b^{(1)},b^{(2)}\}$ where $\lpos(b^{(i)}) = i$ for $i=1,2$. 
Furthermore denote by $\ell_i $ be the $(0,k-1)$-weighted $q$-degree of $b^{(i)}$ for $i=1,2$. Then
\[\ell_1 + \ell_2 = n+k-1\]
or equivalently
\[ \qdeg( b_1^{(1)}) + \qdeg(b_2^{(2)}) = n .\]
\end{lem}
\begin{proof}
By Proposition \ref{prop:samelp} we know that the $q$-degrees of any minimal basis of the interpolation module $\mathfrak M({\bf r})$ have to add up to the same number, hence it is enough to show that they add up to $n+k-1$ for one particular basis. Consider the iterative construction of a minimal basis from Algorithm \ref{alg3}. It is easy to see that the initial basis has weighted $q$-degrees $0$ and $k-1$. Moreover, at each step the $q$-degree of one row is increased by one, whereas the $q$-degree of the other row remains the same. Thus, the sum of the two $q$-degrees is increased by $1$ at each step. Since we get the desired basis of  $\mathfrak M({\bf r})$ at the $n$-th step, the statement follows.
\end{proof}


\begin{lem}\label{lem:l1}
In the setting of Lemma \ref{lem:l1andl2} let $t:=\min\{d_R(\mathbf{ c}, \mathbf{ r})\mid \mathbf{ c} \in C\}$. Then 
\[\ell_2 \leq t+k-1\]
or equivalently $\qdeg(b_2^{(2)}) \leq  t$ .
Furthermore, 
 \[\ell_1 =\qdeg(b_1^{(1)}) \geq n-t .\]
\end{lem}
\begin{proof}
Let $m(x) \in \Lp$ be the message polynomial corresponding to the codeword $\bf c$. Then by Theorem \ref{thm:main}, there exist $D(x)\in\Lp$ of $q$-degree $t$ such that $f:=[\: D(m(x)) \;\; D(x) \:] $ is an element of the interpolation module with leading position $2$. By the PLM property from Theorem \ref{thm:PLM} we know that $\lm (f) = \lm (a(x) \circ b^{(2)})$ for some $a(x) \in \Lp$, i.e.\ $\lm (f) \geq \lm ( b^{(2)})$, which implies the first statement since the leading positions of both elements are $2$.
We know from Lemma \ref{lem:l1andl2} that $\ell_1 + \ell_2 = n+k-1$, i.e.\
\[\ell_1 = n+k-1-\ell_2 \geq n-t.\]
\end{proof}

\begin{lem}
In the previous setting, if $t\leq (n-k)/2$, then $\ell_2 = t+k-1$, or equivalently $\qdeg(b_2^{(2)}) = t$, and $\ell_1 = n-t$.
\end{lem}
\begin{proof}
We know from Lemma \ref{lem:l1} that 
\[ \ell_1 \geq n-t \geq  \frac{n+k}{2}  .\]
Since the vector corresponding to the closest codeword has $(0,k-1)$-weighted $q$-degree $t+k-1< \ell_1$, this vector is $a(x)\circ  b^{(2)}$ for some $a(x) \in \Lp$. But if the divisibility requirement is fulfilled for $a(x)\circ  b^{(2)}$, then it must also be fulfilled for $ b^{(2)}$. Hence  $ b^{(2)}$ must correspond to the closest codeword and is thus of weighted $q$-degree $t+k-1$ (by Theorem \ref{thm:main}). Note that the last step could also be justified by using the fact that we are within the unique decoding radius.
From Lemma \ref{lem:l1andl2} we then get that $\ell_1 = n+k-1-\ell_2 =n-t$.
\end{proof}

\begin{cor}\label{cor:poly}
It follows that, if the received word is within the unique decoding radius, Algorithm \ref{alg1} only performs one loop and hence only one symbolic division to find the message polynomial corresponding to the unique closest codeword.
\end{cor}

\subsection{Complexity Analysis}

We will now analyze the computational complexity of the previous algorithms. For this we assume that our elements of $\F_{q^m}$ are stored as elements of $\F_q^m$ with regard to a normal basis $\{\beta, \beta^{[1]}, \beta^{[2]}, \dots , \beta^{[m-1]}\}$ of $\F_{q^m}$ over $\F_q$. Such a normal basis is most suitable, since taking $q$-th powers are represented by cyclic shifts and can hence be neglected in the complexity order analysis (see e.g.\ \cite[Section 3.1]{wa13phd}).

We will start with the easier task of analyzing Algorithms \ref{alg2} and \ref{alg3}, before we 
 derive the overall decoding complexity, which is the complexity of Algorithm \ref{alg1}, in terms of the rank distance of the closest codeword to the received word. 

Throughout this subsection we will use the notation of the previous subsection, i.e.\ $b^{(1)}, b^{(2)}\in \Lp$ form a minimal basis of the interpolation module $\mathfrak{M} (\bf r )$ and $\ell_i$ is the $(0,k-1)$-weighted $q$-degree of $b^{(i)}$ for $i=1,2$. Moreover, $\lpos(b^{(i)})=i$ for $i=1,2$. $\Pi_{\bf g}(x)$ is the annihilator polynomial for the generators of the code $g_1,\dots,g_n\in \F_{q^m}$ and $\Lambda_{\bf r, \bf g}(x)$ is the $q$-Lagrange polynomial with respect to the received vector $\mathbf{r} \in \F_{q^m}^n$.

\begin{lem}
$\Pi_{\bf g}(x)$ can be computed with at most $\mathcal{O}_{q^m}(n^2)$ operations. 
\end{lem}
\begin{proof}
Consider the iterative construction from Proposition \ref{prop:Lagrec}. 
At step $i$ we need to compute $\Pi_i(x) = (x^q - g_i^{q-1}x)\circ \Pi_{i-1}(x)$. We compute a $(q-1)$-th power of $g_i$, which can be done with a $q$-th power and one division. Moreover, we need to multiply $\Pi_{i-1}(x)$ with this power, which needs at most $i$ operations since $\qdeg(\Pi_{i-1}(x))=i-1$. Similarly we need to take the $q$-th power of all terms of $\Pi_{i-1}(x)$. The last step is to take the difference of the two resulting polynomials (where one has $q$-degree $i$ and the other $i+1$), hence at most $i+1$ operations. Since $i$ is upper bounded by $n$, we get an upper bound of $\mathcal{O}_{q^m}(n)$ operations at each step. Since there are $n$ steps, the overall complexity is upper bounded by $\mathcal{O}_{q^m}(n^2 )$.
\end{proof}

\begin{rem}
Since $\Pi_{\bf g}(x)$ does not depend on the received word, we can precompute and store it. In the following we assume that we precomputed $\Pi_{<g_1,\dots,g_i>}(x)$ for $i=1,\dots,n$ since we will need all of them for the computation of the $q$-Lagrange polynomial.
\end{rem}

\begin{lem}
$\Lambda_{\bf r, \bf g}(x)$ can be computed with at most $\mathcal{O}_{q^m}(n^2)$ operations.
\end{lem}
\begin{proof}
Consider again the iterative construction from Proposition \ref{prop:Lagrec}. 
At step $i$ we need to compute $\Lambda_i(x) = \Lambda_{i-1}(x) - \frac{\Lambda_{i-1}(g_i)-r_i}{\Pi_{i-1}(g_i)} \Pi_{i-1}(x)$. We need to compute two evaluations of polynomials of $q$-degree at most $i$, whose complexity is at most $\mathcal{O}_{q^m}(i)$, and a negligible division and difference. Moreover, we need to take the difference of the two polynomials, which is in the order of at most $\mathcal{O}_{q^m}(i)$ (because of the degrees). The $q$-degrees are at most $n$, hence for each step we need at most $\mathcal{O}_{q^m}(n)$ operations. Since we have $n$ steps, we get the desired complexity order.
\end{proof}
Note that computing $\Pi_{\bf g}(x)$ and $\Lambda_{\bf r, \bf g}(x)$ ad hoc is much more expensive than using the iterative definition, which is why we used the method from Proposition \ref{prop:Lagrec}.


\begin{prop}
Algorithm \ref{alg2} has computational complexity order $\mathcal{O}_{q^m}(n^2)$.
\end{prop}
\begin{proof}
Once we have the basis of the interpolation module, i.e.\ after computing $\Lambda_{\bf r, \bf g}(x)$, the computation of the minimal basis consists of a linearized extended Euclidean algorithm (EEA). The complexity order of the EEA is given by the square of the larger $q$-degree of the two linearized input polynomials (see \cite[Section IV.A]{ga08p}). Hence in our case the order is upper bounded by $\mathcal{O}_{q^m}(n^2)$. 
\end{proof}


\begin{prop}
Algorithm \ref{alg3} has computational complexity order $\mathcal{O}_{q^m}(n^2)$.
\end{prop}
\begin{proof}
For the iterative computation of the minimal basis from Algorithm \ref{alg3} we need $n$ steps. In each step we need some polynomial evaluations and differences to compute $\Delta_i$ and $\Gamma_i$, which needs $\mathcal{O}_{q^m}(n)$ operations (similarly to before). Moreover, we need to multiply a linearized polynomial of $q$-degree at most $n$ by a scalar, which also needs $\mathcal{O}_{q^m}(n)$ operations. Similarly, the last step is the composition with $(x^q-g_i^{q-1} x)$, which is analogous to one step in the computation of $\Pi_{\bf g}(x)$ and is hence in the order of $\mathcal{O}_{q^m}(n)$. Overall we get an upper bound on the complexity of $\mathcal{O}_{q^m}(n^2)$.
\end{proof}

Note that the computational complexity orders of the EEA Algorithm \ref{alg2} and the iterative method Algorithm \ref{alg3} are the same. However, it is shown in \cite{ga08p} that the actual number of operations used by the iterative method is less than the number of operations used by the EEA. 


We now determine the computational complexity order of the overall decoding Algorithm \ref{alg1} which uses the parametrization. The next theorem shows that the computational complexity order of Algorithm \ref{alg1} is exponential if and only if $t$ is greater than the unique decoding radius, i.e.\ if and only if $t>(n-k)/2$.

\begin{thm}\label{overallcompl}
Let $t$ be the rank distance between the received word $\bf r$ and the closest codeword $\mathbf{c} \in C$. 
The complexity order of Algorithm \ref{alg1}, using Algorithm \ref{alg2} or Algorithm \ref{alg3} for the computation of the minimal basis, is upper bounded by
\[\mathcal{O}_{q^m}(  q^{m(2t+k-n)}(t+k)^2 + n^2   )\]
Furthermore, for $t\geq n-k$, this complexity order is at most
\[\mathcal{O}_{q^m}((q^{m(2t+k-n)} +1) n^2   )  .\]
\end{thm}
\begin{proof}
The complexity is dominated by the number of different $\beta(x), \gamma(x)$ we need to consider in Algorithm \ref{alg1}. We know that $j$ runs from $0$ to $t-\ell_2+k-1$ to find the solutions at distance $t$. For the largest value of $j$ the parametrization considers all $\beta(x)\in \Lp$ with $\qdeg(\beta(x))\leq t-\ell_1 +k-1$ and all monic $\gamma(x) \in \Lp$ with $\qdeg(\gamma(x))\leq t-\ell_2 +k-1$. Hence we get 
$$q^{m( t-\ell_1 +k)} q^{m( t-\ell_2 +k-1)} = q^{m( 2t-(\ell_1+ \ell_2) +2k-1)} = q^{m( 2t +k -n)}   $$ 
possible pairs (where we used Lemma \ref{lem:l1andl2} for the last equality). For each such pair $\beta(x),\gamma(x)$ the symbolic division algorithm for linearized polynomials is executed, which has a complexity order of square of the larger $q$-degree of the two polynomials (see e.g.\ \cite{ko08,wa13phd}). The degrees of the respective polynomials in the division algorithm are at most $t+k$, hence the symbolic division has order at most $\mathcal{O}_{q^m}((t+k)^2)$. Now also taking into account the complexity order of Algorithm \ref{alg2} or Algorithm \ref{alg3}, which is performed only once at the beginning of the algorithm, we conclude that the first statement holds. In fact, the degrees of the two polynomials are also upper bounded by $n$, so that the second statement follows.

\end{proof}

Note that it was already shown in  Corollary \ref{cor:poly} that only one loop with one symbolic division needs to be executed if $ t \leq (n-k)/2$, hence the complexity order in this case is given by $\mathcal{O}_{q^m}(n^2)$ (which is the complexity order of Algorithm \ref{alg2} as well as Algorithm \ref{alg3}).

\vspace{0.3cm}

In the following we compare the complexity order of Algorithm \ref{alg1} to the complexity order of  a chase list decoding algorithm. The chase list decoding algorithm is a list decoding method based on the following fact: If we receive ${\bf r}={\bf c}+{\bf e} \in \F_{q^m}^n$ with $\rk({\bf e})=t > (d-1)/2$, then there exists a ${\bf e'} \in \F_{q^m}^n$ of rank $t-(d-1)/2$ such that $\rk({\bf e-e'})=(d-1)/2$. The chase is to try all possible ${\bf e'}$ and use a unique decoder on
$\bf r-e'$. This will find all codewords whose rank distance to the received word ${\bf r}$ equals $t$ or less. 

For a complexity analysis of this chase algorithm we need to count how many of these $\bf e'$ there are. For this we count all matrices in
$\F_q^{m\times n}$ of rank $t-(d-1)/2$. By rank decomposition over $\F_q$ an upper bound on the number of these matrices is
\[q^{m(t-\frac{d-1}{2})} q^{(t-\frac{d-1}{2})n} = q^{(t-\frac{d-1}{2})(n+m)} \leq q^{2m(t-\frac{d-1}{2})}=q^{m(2t-n+k)}.\]
Since the unique decoder can be done in $O_{q^m}(n^2)$ (see e.g.\ \cite{ga85a,lo06}), we get an overall complexity order 
of $ O_{q^m}(q^{m(2t-n+k)} n^2)$ for the chase algorithm. 
Hence, when decoding beyond the unique decoding radius, if $t < n-k$, then Theorem \ref{overallcompl} shows that the complexity order of our list decoding algorithm is less than the complexity order of the chase algorithm. However, for $t \geq n-k$ our list decoding algorithm has the same complexity order as the chase algorithm.

Note that both our algorithm and the chase algorithm perform better than an exhaustive search list decoding algorithm (i.e.\ computing the distance between the received word and every codeword), as long as the decoding radius $t$ is less than $n/2$.

\section{Conclusions}\label{sec:conclusion}

In this paper we used a parametrization approach for decoding Gabidulin codes with respect to the rank metric. Our main result is that we use this algorithm to compute a list of message polynomials that correspond to {\em all} codewords that are closest to a given received word. Thus we do not prescribe a fixed decoding radius as in a complete list decoder. However, our result can straightforwardly be extended to a complete list decoder with prescribed radius $t^*$ by simply applying the search through the parametrization with increasing $t$ until $t=t^*$. 

We summarized some results on modules over the ring of linearized polynomials and emphasized the Predictable Leading Monomial (PLM) property for minimal bases of these modules as a key property in the context of an interpolation module for a given Gabidulin code and a received word. The decoding algorithm and the parametrization were set within this interpolation module, using the PLM property as a key ingredient.
 
To compute a minimal basis of the interpolation module we presented two algorithms -- the first coincides with the extended Euclidean algorithm for linearized polynomials; the second is an iterative algorithm with simple update steps. These algorithms are similar to other known algorithms.
In fact, our extended Euclidean algorithm is similar to the Gao-type algorithm by Wachter-Zeh \cite{wa13phd}, whilst our iterative algorithm coincides with the Welch-Berlekamp type algorithm given by Loidreau \cite{lo06}. The algorithms of \cite{lo06,wa13phd} decode only within the unique decoding radius. The first main contribution of our paper is the recognition that these algorithms 
actually compute a minimal basis of the interpolation module which can then be used to  decode beyond the unique decoding radius. We then showed how to do this, as our second main contribution. 

Finally we gave a complexity analysis, showing that our algorithm has polynomial computational complexity if and only if the received word is within the unique decoding radius. Beyond the unique decoding radius, the complexity of our algorithm is better than the complexity of a chase list decoding algorithm, if the decoding radius is less than $n-k$. If the decoding radius is greater or equal to $n-k$ then our algorithm's complexity order is on par with the complexity order of the chase algorithm. 

In future work we intend to use the parametrization introduced in this paper to tackle the open question of the existence of polynomial size lists for Gabidulin codes for decoding radii between the unique decoding and the Johnson radius. In particular, we aim to investigate whether there are parameter sets for which the list size \emph{is} polynomial; and if so, derive a polynomial time list decoding algorithm for Gabidulin codes.

\section{Acknowledgment}
We would like to thank the anonymous reviewers for helpful comments and for providing additional references.
\bibliographystyle{plain}
\bibliographystyle{spmpsci}
\bibliography{margreta_anna-lena}

\end{document}